\documentclass{article}%
\usepackage{amsfonts}
\usepackage{amsmath}
\usepackage{amssymb}
\usepackage{graphicx}%
\providecommand{\U}[1]{\protect\rule{.1in}{.1in}}
\newtheorem{theorem}{Theorem}

\newtheorem{conclusion}[theorem]{Conclusion}

\newtheorem{example}[theorem]{Example}

\newtheorem{problem}[theorem]{Problem}
\newtheorem{proposition}[theorem]{Proposition}
\newtheorem{remark}[theorem]{Remark}

\newenvironment{proof}[1][Proof]{\noindent\textbf{#1.} }{\ \rule{0.5em}{0.5em}}
\begin{document}

\title{Linear stochastic stability analysis of nonlinear systems. Parametric
destabilization of the wave propagation.   }
\author{Gy\"{o}rgy Steinbrecher\\EURATOM-MEdC, University of Craiova,  Str. A. I. Cuza 13,\\Craiova --200585, Romania
\and Xavier Garbet,\\EURATOM-CEA, IRFM,  F-13108 Saint Paul Lez Durance, France.}
\maketitle

\begin{abstract}
\ Straightforward method for the derivation of linearized version of
stochastic stability analysis of the nonlinear differential equations is
presented. Methods for the study of large time behavior of the moments are
exposed. These general methods are applied to the study of the stochastic
destabilization of the Langmuir waves in plasma.  

\end{abstract}

\section{\bigskip INTRODUCTION}

The influence of the external noise on the stability of the linearized
dynamical systems was the subject of numerous studies on finite dimensional
models, studied both in physical or mathematical literature \cite{Hoff}%
-\cite{SXG}.  In the following we extend part of the formalism to dynamical
systems described by partial differential or integro-differential stochastic equations.

The main motivation of this study is related to the sensitivity of the onset
to unstable regime, under the effect of random noise, observed in hydrodynamic
turbulence \cite{Hoff}, the effect of stochastic destabilization of the stable
regime of simplest Hamiltonian systems, like linear undamped $LC$ system with
random capacitance \cite{Lindberg}, \cite{Gitterman}. \ The main dynamical
system of interest for the controlled thermonuclear fusion studies with
magnetic confinement are described the kinetic equations. Before the study of
these more complete models, we expose the main physical and mathematical
principles on a more simple model: of Langmuir waves in turbulent background.\ 

The large part of the mathematical aspects exposed here, at least those
applied to finite dimensional systems, can be found in the mathematical
literature. An example  of application to large discrete systems, perturbed by
space and temporal white noise, can be found in the works \cite{KN1}%
-\cite{KN4}. 

The derivation of the equations for the stochastic linearized system from the
stochastic nonlinear one as well as the application to Langmuir waves, is new. 

\subsection{Physical model }

The main class of problems studied in this article are related to the
influence of the external noise on the linear stability of steady state of
general, non linear physical systems. We suppose that the system is described
by a set of general nonlinear partial differential, or integro-differential
equations, that contains a stochastic perturbation, modelled by a random
field. \ 

In all of this part of the study the random field will be modelled by temporal
white noise with general spatial correlations.

Instead of establishing a general abstract framework, in this article we
illustrate the method by an example. The first interesting physical model is
non the linear hyperbolic evolution equation, that describe the propagation of
nonlinear \ waves in plasma, under the effects of the external noise
$n_{\omega}(\mathbf{x},t)$.
\begin{equation}
\frac{\partial^{2}\Phi_{\omega}(\mathbf{x},t)}{\partial t^{2}}=\Delta
\Phi_{\omega}-F(\Phi_{\omega}(\mathbf{x},t),n_{\omega}(\mathbf{x}%
,t))\label{i1}%
\end{equation}

Here $F(u,v)$ is a nonlinear function, having near $u=v=0$ the linearized form
$F(u,v)\underset{u,v\rightarrow0}{=}\alpha u+\beta v+\mathcal{O}(u^{2}+v^{2}%
)$, with $u>0$. \ \ We are interested in the linear stability of the solution
in the deterministic case $\Phi(\mathbf{x},t)=0$, under the effect of the
noise $n(\mathbf{x},t)$ , represented in fact by a random field whose
statistical properties are supposed to be known. In particular we suppose that
the random field is homogenous and stationary. \

We illustrate now the general recipe to reduce the linear stochastic stability
analysis of non linear systems to standard form. Similar to the deterministic
case we expand $F(\Phi,n)$ in the variable $\Phi$ up to order $2$ term to
obtain%
\begin{equation}
F(\Phi,n_{\omega})=\alpha\Phi+M_{\omega}(\mathbf{x},t)\Phi+a_{\omega
}(\mathbf{x},t)+\mathcal{O}(\Phi^{2})\label{i2}%
\end{equation}

where $\alpha>0$ is constant, and the multiplicative and additive noise term
are given exactly by
\begin{align}
M_{\omega}(\mathbf{x},t)  & =\ \left[  \frac{\partial\left[  F(\Phi,n_{\omega
}(\mathbf{x},t))-F(\Phi,0)\right]  }{\partial\Phi}\right]  _{\Phi=0}%
\label{i3}\\
a_{\omega}(\mathbf{x},t)  & =F(0,n_{\omega}(\mathbf{x},t))\label{i4}%
\end{align}

\bigskip Neglecting the $\mathcal{O}(\Phi^{2})$ term and introducing the
variable \ $P_{\omega}(\mathbf{x},t)=\frac{\partial\Phi_{\omega}%
(\mathbf{x},t)}{\partial t}$  , we obtain the linear system of stochastic
differential equation of the first order form%
\begin{align}
d\Phi_{\omega}(\mathbf{x},t)  & =P_{\omega}dt\label{i5}\\
dP_{\omega}(\mathbf{x},t)  & =(\Delta+\alpha)\Phi dt+dW_{\omega}%
(\mathbf{x},t)\Phi\ +a_{\omega}(\mathbf{x},t)dt\label{i6}\\
dW_{\omega}(\mathbf{x},t)  & =M_{\omega}(\mathbf{x},t)dt\label{i7}%
\end{align}

Without loss of generality we can consider the case
\[
\mathbb{E}_{\omega}\left[  M_{\omega}(\mathbf{x},t)\right]  =0
\]
In the following we will approximate $M_{\omega}(\mathbf{x},t)$ by a temporal
white noise, consequently we have%
\begin{align*}
\mathbb{E}_{\omega}\left[  W_{\omega}(\mathbf{x},t)\right]    & =0\\
\mathbb{E}_{\omega}\left[  W_{\omega}(\mathbf{x},t)W_{\omega}(\mathbf{x}%
^{\prime},t^{\prime})\right]    & =\min(t,t^{\prime})C(\mathbf{x,x}^{\prime})
\end{align*}
Here $C(\mathbf{x,x}^{\prime})$ is the spatial correlation function of the
random field, related to the random multiplicative terms. By using $L_{p}$
estimates like in the work \cite{SGXGBW}, it can be proven a result similar to
\cite{SGBW}, that under general conditions the additive term does not provide
exponential destabilization of the moments of the random fields $\Phi_{\omega
},P_{\omega}$ . Consequently we consider in the continuation only the
homogenous system, while we are interested only in the large time behavior of
the linearized system Eqs.(\ref{i5}-\ref{i7}).

The previous system of stochastic differential equation can be formally
rewritten as general linear stochastic differential equation in a suitable
Banach space as follows%
\begin{equation}
d\mathbf{Y}_{\omega}(t)=\left[  \widehat{A}dt+d\widehat{B}_{\omega}(t)\right]
\mathbf{Y}_{\omega}(t)\label{i9}%
\end{equation}

where
\begin{align*}
\mathbf{Y}_{\omega}  & =%
\begin{pmatrix}
\Phi_{\omega}\\
P_{\omega}%
\end{pmatrix}
\\
\widehat{A}  & =%
\begin{pmatrix}
0 & 1\\
\Delta+\alpha & 0
\end{pmatrix}
\\
d\widehat{B}_{\omega}(t)  & =%
\begin{pmatrix}
0 & 0\\
W_{\omega}(\mathbf{x},t) & 0
\end{pmatrix}
\end{align*}

Similar linearization procedure can be performed for the system of drift
kinetic equations, describing various species of charged particles
$(q_{\alpha},m_{\alpha})$ , $1\leq\alpha\leq A$in a strong constant magnetic
field $\mathbf{B}$. The starting nonlinear system of equations for the
functions $f_{a}(\mathbf{x,v},t)$ can be written as follows%
\begin{align*}
\frac{\partial f_{a}(\mathbf{x,v},t)}{\partial t}  & =-v_{\parallel}%
\frac{\partial f_{a}}{\partial z}-\frac{q_{\alpha}}{m_{\alpha}}E_{\parallel
}\frac{\partial f_{a}}{\partial v_{\parallel}}-\frac{\mathbf{E\times B}}%
{B^{2}}\cdot\nabla_{\perp}f_{a}\\
\varepsilon_{0}\Delta_{\mathbf{x}}f_{a}  & =%
{\displaystyle\sum_{\alpha=1}^{A}}
q_{\alpha}%
{\displaystyle\int}
f_{a}(\mathbf{x,v},t)d^{3}\mathbf{v}%
\end{align*}

The linearization of the equations starts with the ansatz
\[
f_{a}(\mathbf{x,v},t)=\delta f_{\omega,a}+\left[  n_{0}(\mathbf{x})+\delta
n_{\omega}(\mathbf{x,t})\right]  \exp\left[  -v^{2}\left(  \beta
_{0}(\mathbf{x})+\delta\beta_{\omega}(\mathbf{x},t)\right)  \right]
\]

around the equilibrium density and inverse temperature profiles $n_{0}%
(\mathbf{x}),\beta_{0}(\mathbf{x})$ . The stochastic effects are induced by
the temperature and density fluctuations $\delta\beta_{\omega}(\mathbf{x}%
,t),\delta n_{\omega}(\mathbf{x,t})$ . By performing the same linearization
procedure we obtain a system of linear stochastic differential equations of
the form Eq.(\ref{i9}) where
\[
\mathbf{Y}_{\omega}=%
\begin{pmatrix}
\delta f_{\omega,1}(\mathbf{x,v},t)\\
\delta f_{\omega,2}(\mathbf{x,v},t)\\
\cdots\\
\delta f_{\omega,A}(\mathbf{x,v},t)
\end{pmatrix}
\]

\subsection{Mathematical background}

We will study now the  correct formulation of the stability problem of the
solutions of Eq.(\ref{i9}), and describe some general methods to solve these
problems. 

Consider a general linear evolution equation in a complex Banach space
$\mathbf{V}$ with the norm $\left\Vert .\right\Vert $ .
\begin{equation}
\frac{d\mathbf{x}(t)}{dt}=\widehat{A}\mathbf{x}(t);~\mathbf{x}(t)\in
\mathbf{V};\label{10}%
\end{equation}

where $\widehat{A}$ is a linear operator with dense domain in $\mathbf{V}$
\ and range in $\mathbf{V}.$We will denote by $\mathbf{V}^{\ast}$ the dual
space of $\mathbf{V}$, where the transpose $\widehat{A}^{T}$ \ of $\widehat
{A}$ acts.

$~$An equivalent form, written in components ( or reducing our study to the
discretized form of the partial differential of integro-differential
equations) is%
\begin{equation}
\frac{dx^{i}(t)}{dt}=A_{j}^{i}x^{j}(t)\label{20}%
\end{equation}
$\ $\ where summation on repeated indices is understood.

Typical examples of interest will be the linearized equations used in the
deterministic linear stability studies.

In the framework of the deterministic linear stability studies, in the finite
dimensional case, a clear answer to the stability problem is given by the
eigenvalue of $\widehat{A}$ with largest real part. In general case we\ denote
by $sp(\widehat{A})$ the spectrum of $\widehat{A}$ and
\begin{equation}
\lambda=\max\left\{  \operatorname{Re}(z)|z\in sp(\widehat{A})\right\}
\label{30}%
\end{equation}

The real number $\lambda$ is the largest Liapunov exponent. Then for all
$\varepsilon>0$ we have for any solution of Eqs.(\ref{10}, \ref{20})%

\begin{equation}
\underset{t\rightarrow\infty}{\lim}\frac{\left\Vert \mathbf{x}(t)\right\Vert
}{\exp\left[  \left(  \lambda+\varepsilon\right)  t\right]  }=0 \label{40}%
\end{equation}

If $\lambda$ is known, then the question about the stability of null solution
in Eqs.(\ref{10}, \ref{20}) is solved.

Consider now the stochastic generalization of Eqs.(\ref{10}, \ref{20}), by
adding a random, time dependent a linear term, for the sake of simplicity
modelled by an operator -valued Brownian motion . In this end, we consider
$\left(  \Omega,\mathcal{A},\mathcal{F}_{t},P\right)  $ a filtered probability
space and denote by $\widehat{B}_{\omega}(t)$ , with $\omega\in\Omega$, an
operator-valued stochastic process adapted to the filtration $\mathcal{F}_{t}%
$. In the our case $\widehat{B}_{\omega}(t)$ is an operator valued Brownian
motion, or equivalently $\widehat{B}_{\omega}(t)\mathbf{x}$ is an $\mathbf{V}$
valued Brownian motion \cite{Oksendal} for all $\mathbf{x\in V}$.

\begin{remark}
\label{RemIto} We will use everywhere the It\^{o} formalism.
\end{remark}

We have the following stochastic evolution equation,
\begin{equation}
d\mathbf{x}_{\omega}(t)=\left[  \widehat{A}dt+d\widehat{B}_{\omega}(t)\right]
\mathbf{x}_{\omega}(t)\label{50}%
\end{equation}

or in components ( we use the general relativity index conventions)%
\begin{align}
dx_{\omega}^{i}(t)  &  =\left[  A_{j}^{i}dt+\rho_{j,a}^{i}dw_{\omega}%
^{a}(t)\right]  x_{\omega}^{j}(t)\label{60}\\
1  &  \leq i,j\leq n;\ 1\leq a\leq A \label{6.01}%
\end{align}

\bigskip where $\rho_{j,a}^{i}$ are constants and $w_{\omega}^{a}(t)$ are
independent standard Brownian motions%
\begin{equation}
\mathbb{E}_{\omega}\left[  w_{\omega}^{a}(t)w_{\omega}^{b}(t^{\prime})\right]
=\min(t,t^{\prime})\delta^{a,b} \label{6.02}%
\end{equation}

\bigskip We have the matrix-valued Brownian motions$\widehat{B}_{\omega}(t),$
with components $B_{j,\omega}^{i}(t)$, having the correlation tensor
$C_{j,n}^{i,m}$%
\begin{align}
\mathbb{E}_{\omega}\left[  B_{j,\omega}^{i}(t)B_{n,\omega}^{m}(t^{\prime
})\right]   &  =\min(t,t^{\prime})C_{j,n}^{i,m}\label{61}\\
C_{j,n}^{i,m}  &  =\sum_{a=1}^{A}\rho_{j,a}^{i}\rho_{n,a}^{m} \label{61.1}%
\end{align}

To this correlation tensor we associate a linear operator $\widehat
{\mathbf{C}}$ acting in the tensor product space $\mathbf{V\otimes V}$, so an
equivalent form of Eqs. (\ref{61}, \ref{61.1}) is%
\begin{equation}
\mathbb{E}_{\omega}\left\{  \left[  \widehat{B}_{\omega}(t)\mathbf{x}\right]
\otimes\left[  \widehat{B}_{\omega}(t^{\prime})\mathbf{y}\right]  \right\}
=\min(t,t^{\prime})\widehat{\mathbf{C}}\left(  \mathbf{x}\otimes
\mathbf{y}\right)  ;\mathbf{x}\otimes\mathbf{y\in V\otimes V} \label{61.2}%
\end{equation}

An equivalent form of the Eq.(\ref{61.2}) \ that will be useful is%
\begin{equation}
\mathbb{E}_{\omega}\left\{  \left[  d\widehat{B}_{\omega}(t)\mathbf{x}\right]
\otimes\left[  d\widehat{B}_{\omega}(t)\mathbf{y}\right]  \right\}
=\widehat{\mathbf{C}}\left(  \mathbf{x}\otimes\mathbf{y}\right)  dt
\label{61.3}%
\end{equation}

\begin{remark}
\label{RemPeq1}Observe immediately that exactly due to the It\^{o} definition,
if we denote $\mathbf{Y}_{1}(t)=\mathbb{E}\mathbf{x}_{\omega}(t),$ then we
obtain from Eq.(\ref{50})
\begin{equation}
\frac{d\mathbf{Y}_{1}(t)}{dt}=\widehat{A}\mathbf{Y}_{1}(t) \label{61.4}%
\end{equation}
so apparently the effect of the noise is averaged out. For higher order
moments the equations will contain also the correlation effects of the noise.
\end{remark}

Nevertheless, we can try generalize Eq.(\ref{40}) as follows:

\begin{problem}
\label{ProblemLambda_p} Find the minimal value of $\lambda_{p}$ such that for
$p\geq1$ and all $\varepsilon>0$ we have \bigskip$\ $%
\begin{equation}
\underset{t\rightarrow\infty}{\lim}\frac{\left[  \mathbb{E}_{\omega}\left[
\left\Vert \mathbf{x}_{\omega}(t)\right\Vert ^{p}\right]  \right]  ^{1/p}%
}{\exp\left[  \left(  \lambda_{p}+\varepsilon\right)  t\right]  }=0 \label{70}%
\end{equation}

\end{problem}

Simplest soluble one-dimensional cases, as well as more general results
\cite{Khasminsky}, suggest that in this case the exponent $\lambda$ is
dependent on the exponent $p$ in a non-trivial way: in the case of systems
that in deterministic case are linearly stable, \ it is possible that we have
Eq.(\ref{70}) for small noise intensity and small values of $p$ we have
$\lambda_{p}<0$, hence stability, but surely for sufficiently large $p$ we
have $\ \lambda_{p}>0$ .

\begin{example}
Consider the one dimensional model%
\[
dx_{\omega}(t)=(-adt+\rho dw_{\omega}(t))x_{\omega}(t)
\]
in the It\^{o} formalism, where $w_{\omega}(t)$ is the standard Wiener process
($\mathbb{E}_{\omega}w_{\omega}(t)=0;\mathbb{E}_{\omega}[w_{\omega}(t)]^{2}%
=t$), $a,~\sigma$ are constants. The general solution is%
\[
x_{\omega}(t)=x(0)\exp[-(a+\rho^{2}/2)t+\sigma w_{\omega}(t)]
\]
. Results%
\[
\left[  \mathbb{E}_{\omega}|x_{\omega}(t)|^{p}\right]  ^{1/p}=|x(0)|\exp
(\lambda_{p}t)
\]
with $\lambda_{p}=-a+(p-1)\rho^{2}/2$ .
\end{example}

\begin{conclusion}
In the stochastic case the linear stability is determined by a scale,
$\lambda_{p}\,\ $\ of Liapunov exponents
\end{conclusion}

The function $\lambda_{p}\,$\ was effectively computed only in some special
cases (\cite{SGBW}, \cite{SGXGBW}).

We will see that excepting the small noise, perturbative calculations, in the
physically interesting cases ( partial differential equations,
\ Vlasov-Poisson systems) \ computation of $\lambda_{p}$ for even values of
$p$ is possibly only at the expense of solving eigenvalue problems \ with
spatial dimension increased by a factor $p\,$. \ 

In this moment there are no general method to compute $\lambda_{p}$ , for all
real values of $p$, even for the case when $\widehat{A}$ or $\widehat
{B}_{\omega}(t)$ are $2\times2$ matrices. So we reformulate the Problem(
\ref{ProblemLambda_p}) in a more weak sense as follows

\begin{problem}
\label{ProblemMonomials} For a given degree $p$, find $\lambda_{p}$ such that
for all $\varepsilon>0$ and all monomial $m_{p}(\mathbf{x})$ of degree
$p>1$(i.e. $m_{p}(\mathbf{x})=x^{i_{1}}x^{i_{2}}\cdots x^{i_{n}}$ with
$i_{1}+\cdots+i_{n}=p$ ) . the case $p=1$ is trivial, see Remark(\ref{RemPeq1}%
) ). In the components $x_{i}$ we have
\begin{equation}
\underset{t\rightarrow\infty}{\lim}\frac{\mathbb{E}_{\omega}\left[
m_{p}\left(  \mathbf{x}_{\omega}\left(  t\right)  \right)  \right]  }%
{\exp\left[  p\left(  \lambda_{p}+\varepsilon\right)  t\right]  }=0 \label{80}%
\end{equation}

\end{problem}

In this end we will obtain closed sets of deterministic equations for the
moments%
\begin{equation}
\mathbb{E}_{\omega}\left[  x_{\omega}^{i_{1}}\left(  t\right)  x_{\omega
}^{i_{2}}\left(  t\right)  \cdots x_{\omega}^{i_{n}}\left(  t\right)
\ \right]  \label{90}%
\end{equation}

\section{The deterministic equations for the moments.}

\subsection{The evolution of the mean values}

For simplicity consider first the finite dimensional case, denote in general
$\mathbf{x}=(x^{1},x^{2},...x^{n})$.We have the following

\begin{remark}
\label{BacwardKolmogorov} (Backward Kolmogorov equation) Consider a general
stochastic process $\mathbf{x}_{\omega}(t)$, described by the following SDE
and initial conditions
\begin{align}
dx_{\omega}^{i}(t) &  =V^{i}(\mathbf{x}_{\omega}(t))dt+\sigma_{a}%
^{i}(\mathbf{x}_{\omega}(t)\mathbf{)}dw_{\omega}^{a}(t);1\leq i\leq n;1\leq
a\leq A\label{100}\\
\mathbb{E}_{\omega}\left[  \ \ w_{\omega}^{a}(t)\right]   &  =0;\mathbb{E}%
_{\omega}\left[  \ \ w_{\omega}^{a}(t)w_{\omega}^{b}(t^{\prime})\right]
=\delta^{a,b}\min\left(  t,t^{\prime}\right)  \label{110}\\
\mathbf{x}_{\omega,\mathbf{y}}(0) &  =\mathbf{x}_{0}\label{130}%
\end{align}
For for any differentiable function $F(\mathbf{x})$ we have (for the sake of
simplicity of the notations, the subscript $\omega$ will be omitted)
\begin{align}
dF(\mathbf{x}_{\omega}(t)) &  =\frac{\partial F}{\partial x^{i}}\left[
V^{i}(\mathbf{x}_{\omega}(t))dt+\sigma_{a}^{i}(\mathbf{x}_{\omega
}(t)\mathbf{)}dw_{\omega}^{a}(t)\right]  +\\
&  \frac{1}{2}\frac{\partial^{2}F}{\partial x^{j}\partial x^{i}}%
D^{i,j}(\mathbf{x})dt\nonumber\\
D^{i,j}(\mathbf{x}) &  =%
{\displaystyle\sum_{a=1}^{A}}
\sigma_{a}^{i}(\mathbf{x})\sigma_{a}^{j}(\mathbf{x})\label{140}%
\end{align}
\ We denote%
\begin{equation}
M_{F}(t)=\mathbb{E}_{\omega}\left[  F\left(  \mathbf{x}_{\omega}\left(
t\right)  \right)  \right]  \ \label{150}%
\end{equation}
Then we have
\begin{equation}
\frac{M_{F}(t)}{dt}=\mathbb{E}_{\omega}\left[  \frac{\partial F(\mathbf{x}%
_{\omega}(t))}{\partial x^{i}}V^{i}(\mathbf{x}_{\omega}(t))+\frac{1}{2}%
\frac{\partial^{2}F}{\partial x^{j}\partial x^{i}}D^{i,j}(\mathbf{x}_{\omega
}\left(  t\right)  )\right]  \ \label{160}%
\end{equation}

\end{remark}

It follows from Eqs.(\ref{160}, \ref{140}) that in the case when the
SDE(\ref{100}) is linear, it is possible to obtain a closed set of linear
differential equations \ when $F(\mathbf{x})$ belongs to a finite dimensional
subspace of homogenous polynomials. \ 

We consider now the case when Eq.(\ref{100}) is a linear SDE, like
Eq.(\ref{60}). So \
\begin{align}
V^{i}(\mathbf{x})  &  =A_{j}^{i}x^{j}\label{170}\\
\sigma_{a}^{i}(\mathbf{x)}  &  =\rho_{j,a}^{i}x^{j} \label{180}%
\end{align}

\bigskip We obtain from Eqs.(\ref{140}, \ref{180}, \ref{61.1})
\begin{equation}
D^{i,j}(\mathbf{x})=C_{m,n}^{i,j}~x^{m}~x^{n} \label{190}%
\end{equation}

We obtain from Eqs.(\ref{160}, \ref{170}, \ref{190}) the basic equation for
further development\bigskip%
\begin{equation}
\frac{M_{F}(t)}{dt}=\mathbb{E}_{\omega}\left[  \frac{\partial F(\mathbf{x}%
_{\omega}(t))}{\partial x^{i}}A_{j}^{i}x_{\omega}^{j}(t)+\frac{1}{2}%
\frac{\partial^{2}F}{\partial x^{j}\partial x^{i}}C_{m,n}^{i,j}~x_{\omega}%
^{m}(t)~x_{\omega}^{n}(t)\right]  \label{200}%
\end{equation}

\subsection{Linear equations for the moments of order $N$.}

\subsubsection{Case $N=1.$}

Consider, for more clarity reasons, in Eq.(\ref{200}) $F(\mathbf{x}%
\ )\overset{def}{=}x^{k}$, where $\ k$ is fixed integer index and denote in
general $Y_{1}^{j}(t)=$ $\mathbb{E}_{\omega}\left[  \ x_{\omega}%
^{j}(t)\right]  $ , as in Remark(\ref{RemPeq1}).

\emph{Recall, }$x^{k}$\emph{ is a contravariant vector component, not power. }

Then we obtain Eq.(\ref{61.4}), or
\begin{align}
\frac{dY_{1}^{k}}{dt} &  =A_{j}^{i}Y_{1}^{j}(t)\label{210}\\
\frac{d\mathbf{Y}_{1}(t)}{dt} &  =\widehat{A}\mathbf{Y}_{1}(t)
\end{align}

The corresponding stability problem is the classical, deterministic

\begin{problem}
\label{ProblemDeterministicStability} Compute the eigenvalues, \ left and
right eigenvectors in the eigenvalue problems
\begin{align*}
\ \lambda\mathbf{u} &  =\widehat{A}\mathbf{u;u\in V;}\\
\lambda\mathbf{v} &  =\widehat{A}^{T}\mathbf{v;~v\in V}^{\ast};\mathbf{v~}%
is~left~eigenvector
\end{align*}

\end{problem}

\subsubsection{Case $N=2$}

Consider in Eq.(\ref{200}) $F(\mathbf{x}\ )\overset{def}{=}x^{k}x^{l}$ and
denote $Y_{2}^{k,l}(t)=$ $\mathbb{E}_{\omega}\left[  \ x_{\omega}%
^{k}(t)x_{\omega}^{l}(t)\right]  $. Then we obtain%
\begin{equation}
\frac{dY_{2}^{k,l}}{dt}=A_{j}^{k}Y_{2}^{j,l}(t)+A_{j}^{l}Y_{2}^{k,j}%
(t)+C_{m,n}^{k.l}Y_{2}^{m,n}(t) \label{220}%
\end{equation}

In condensed notation, taking into account that to $Y_{2}^{j,l}(t)$ we can
associate a tensor $\mathbf{Y}_{2}(t)\in\mathbf{V\otimes V}$ and to
$C_{m,n}^{k.l}$ the operator $\widehat{\mathbf{C}}$ (see Eq.(\ref{61.3})) we
have the equivalent form of Eq.(\ref{220})$_{{}}$%
\begin{equation}
\frac{d\mathbf{Y}_{2}(t)}{dt}=\left[  \widehat{A}\mathbf{\otimes}%
\widehat{\boldsymbol{1}}+\widehat{\boldsymbol{1}}\mathbf{\otimes}\widehat
{A}\mathbf{+}\widehat{\mathbf{C}}\right]  \mathbf{Y}_{2}(t) \label{230}%
\end{equation}

\subsubsection{Case $N=3$}

Similarly we consider $F(\mathbf{x}\ )\overset{def}{=}x^{k}x^{l}x^{i}$ and
denote $Y_{3}^{k,l,i}(t)=$ $\mathbb{E}_{\omega}\left[  \ x_{\omega}%
^{k}(t)x_{\omega}^{l}(t)x_{\omega}^{i}(t)\right]  $ . Then we obtain%
\begin{align}
\frac{dY_{3}^{k,l,i}}{dt}  &  =A_{j}^{k}Y_{3}^{j,l.i}\ +A_{j}^{l}Y_{3}%
^{k,j,i}\ +A_{j}^{i}Y_{3}^{k,l,j}\ +\label{240}\\
&  C_{m,n}^{k.l}Y_{3}^{m,n,i}\ +C_{m,n}^{k.i}Y_{3}^{m,l,n}+C_{m,n}^{l,i}%
Y_{3}^{k,m,n}\nonumber
\end{align}

In order to write in a more compact form, we introduce $\mathbf{Y}_{3}%
(t)\in\mathbf{V\otimes V\otimes V}$ with components $Y_{3}^{k,l,i}$ and the
operators \ $\widehat{\mathbf{C}}_{1,2}$, $\widehat{\mathbf{C}}_{1,3}$, and
$\widehat{\mathbf{C}}_{2,3},$ acting in $\mathbf{V\otimes V\otimes V}$
according to the \ last terms in Eq. (\ref{240}), and we obtain%

\begin{align}
\frac{d\mathbf{Y}_{3}(t)}{dt}  &  =(\widehat{A}\mathbf{\otimes}\widehat
{\boldsymbol{1}}\mathbf{\otimes}\widehat{\boldsymbol{1}}+\widehat
{\boldsymbol{1}}\mathbf{\otimes}\widehat{A}\mathbf{\otimes}\widehat
{\boldsymbol{1}}\mathbf{+}\widehat{\boldsymbol{1}}\mathbf{\otimes
\widehat{\boldsymbol{1}}\otimes}\widehat{A}+\label{250}\\
&  \widehat{\mathbf{C}}_{1,2}+\widehat{\mathbf{C}}_{1,3}+\widehat{\mathbf{C}%
}_{2,3})\mathbf{Y}_{3}(t)\nonumber
\end{align}

\subsubsection{The general $m$ case}

With the choice $F(\mathbf{x}\ )\overset{def}{=}%
{\displaystyle\prod_{j=1}^{m}}
x^{k_{j}}$ and notation $Y_{m}^{k_{1},\cdots,k_{m}}(t)=$ $\mathbb{E}_{\omega
}\left[  \
{\displaystyle\prod_{j=1}^{m}}
x_{\omega}^{k_{j}}(t)\right]  $, from Eq.(\ref{200}) results%
\begin{align}
\frac{dY_{m}^{k_{1},\cdots,k_{m}}}{dt}  &  =A_{j}^{k_{1}}Y_{m}^{j,k_{2}%
,\cdots,k_{m}}\ +A_{j}^{k_{2}}Y_{m}^{k_{1},j,\cdots,k_{m}}\ +A_{j}^{k_{m}%
}Y_{m}^{k_{1},\cdots,j}\ +\label{260}\\
&
{\displaystyle\sum_{a<b}}
C_{j,l}^{k_{a}.k_{b}}Y_{m}^{k_{1},\cdots k_{a-1},j,\cdots,k_{b-1}%
,l,\cdots,k_{m}}\ \nonumber
\end{align}

Similarly we have the compact form%
\begin{align}
\frac{d\mathbf{Y}_{m}(t)}{dt}  &  =[\widehat{A}\mathbf{\otimes}\left(
\widehat{\boldsymbol{1}}\right)  ^{\otimes(m-1)}+\widehat{\boldsymbol{1}%
}\mathbf{\otimes}\widehat{A}\mathbf{\otimes\left(  \widehat{\boldsymbol{1}%
}\right)  ^{\otimes(m-2)}+}\label{270}\\
&  \mathbf{\cdots+\left(  \widehat{\boldsymbol{1}}\right)  ^{\otimes
(m-1)}\otimes}\widehat{A}+%
{\displaystyle\sum_{a<b}}
\widehat{\mathbf{C}}_{a,b}]\mathbf{Y}_{m}(t)\nonumber\\
&  \overset{def.}{=}\widehat{U}_{m}(\widehat{\mathbf{C}})\mathbf{Y}_{m}(t)
\label{271}%
\end{align}

where the action of the operator $\widehat{\mathbf{C}}_{a,b}$ in the tensor
product space $\mathbf{V}^{\mathbf{\otimes m}}$ is defined according to the
corresponding term in Eq.(\ref{260}). Then the large time behavior of
$\mathbf{Y}_{m}(t)$ is dominated by the term $\exp(\lambda_{m}t)\mathbf{Z}%
_{m}$ where $\lambda_{m}$ , in the finite dimensional case, is the eigenvalue
with largest real part in
\begin{equation}
\lambda\mathbf{Z}_{m}=\widehat{U}_{m}(\widehat{\mathbf{C}})\mathbf{Z}_{m}
\label{273}%
\end{equation}

For application of the perturbative analysis in the limit of weak noise
intensity (in Eq.(\ref{270}) the contribution of $\widehat{\mathbf{C}}_{a,b}$
is small) it \ is useful to remark that, Eq.(\ref{273}) can be put in the form%
\begin{equation}
\lambda\mathbf{Y}_{m}=\left(  \widehat{M}_{m}\ +\delta\widehat{M}_{m}\right)
\mathbf{Y}_{m} \label{274}%
\end{equation}

where the unperturbed term is
\begin{align}
\widehat{M}_{m}  &  =\widehat{A}\mathbf{\otimes}\left(  \widehat
{\boldsymbol{1}}\right)  ^{\otimes(m-1)}+\widehat{\boldsymbol{1}%
}\mathbf{\otimes}\widehat{A}\mathbf{\otimes\left(  \widehat{\boldsymbol{1}%
}\right)  ^{\otimes(m-2)}+}\label{275}\\
&  \mathbf{\cdots+\left(  \widehat{\boldsymbol{1}}\right)  ^{\otimes
(m-1)}\otimes}\widehat{A}\nonumber
\end{align}

and the perturbation is%
\begin{equation}
\delta\widehat{M}_{m}=%
{\displaystyle\sum_{a<b}}
\widehat{\mathbf{C}}_{a,b} \label{276}%
\end{equation}

\begin{remark}
If the deterministic stability problem \ref{ProblemDeterministicStability} is
solved then \ the spectrum (including eigenvalues), , left and right
eigenvector of $\widehat{M}_{m}$ \ from Eq. (\ref{275})can be computed immediately.
\end{remark}

\subsubsection{Perturbative approach.}

\paragraph{Perturbation of general, not necessary self adjoint operators.}

These results are generalizations of the perturbations methods from quantum
mechanics to the case when the operators are not necessary self adjoint.

The formalism exposed above is useful to derive perturbative methods,to solve
the eigenvalue problems like from Eq.(\ref{273}).Consider a general complex
linear vector space $\mathbf{U}$ Suppose that the deterministic linear
stability analysis was performed, so we know the eigenvalues, $\nu_{k}$, as
well as the right eigenvectors $\mathbf{u}_{R,k}\in\mathbf{U}$ and left
eigenvectors $\mathbf{v}_{L,k}\in\mathbf{U}^{\ast}$ of the operator
$\widehat{M}\ $. \ The left eigenvectors, in matrix notations, can be
associated to the eigenvectors of the transposed $\widehat{M}^{T}$ . Denote
the linear function on $\mathbf{U}$ associated to $\mathbf{x}^{\ast}%
\in\mathbf{U}^{\ast}$ by $\mathbf{x\rightarrow}\left\langle \mathbf{x}^{\ast
},\mathbf{x}\right\rangle $ . The transposed $\widehat{M}^{T}$ is defined by
$\ $%
\begin{equation}
\left\langle \mathbf{x}^{\ast},\widehat{M}\mathbf{x}\right\rangle
=\left\langle \widehat{M}^{T}\mathbf{x}^{\ast},\mathbf{x}\right\rangle
\label{pert0}%
\end{equation}

The spectrum of $\widehat{M}$ and $\widehat{M}^{T}$ identical in very general
case. In particular in finite dimensions it results from the property of
\ determinants. Much of our results are explained in the finite dimensional
case and we do not discuss the specific problems related to continuos spectrum
and unbounded operators . \ 

\subparagraph{Simple eigenvalue}

\ \ So to each eigenvalue $\nu_{k}$ corresponds the right eigenvectors
$\mathbf{\ u}_{R,k}\in\mathbf{U}$, \ and the left eigenvectors $\mathbf{v}%
_{L,k}\in\mathbf{U}^{\ast}$ . We have $\mathbf{\ }$%
\begin{align}
\widehat{M}\mathbf{u}_{R,k}  &  =\nu_{k}\mathbf{u}_{R,k}\label{pert1}\\
\widehat{M}^{T}\mathbf{v}_{L,k}  &  =\nu_{k}\mathbf{v}_{L,k} \label{pert2}%
\end{align}

\begin{proposition}
\label{propositionSimpleEigenvaluePerturbation} The eigenvalues $\nu
_{k}+\delta\nu_{k}$ of the perturbed operator $\widehat{M}+\delta\widehat{M}$
can be computed in first approximation order by
\begin{equation}
\delta\nu_{k}=\frac{\left\langle \mathbf{v}_{L,k},\delta\widehat{M}%
\mathbf{u}_{R,k}\right\rangle }{\left\langle \mathbf{v}_{L,k},\mathbf{u}%
_{R,k}\right\rangle } \label{pert3}%
\end{equation}

\end{proposition}

\begin{proof}
The perturbed right eigenvector will be represented as $\mathbf{u}%
_{R,k}+\delta\mathbf{u}_{R,k}$.Results%
\begin{equation}
\left(  \widehat{M}+\delta\widehat{M}\right)  \left(  \mathbf{u}_{R,k}%
+\delta\mathbf{u}_{R,k}\right)  =\left(  \nu_{k}+\delta\nu_{k}\right)  \left(
\mathbf{u}_{R,k}+\delta\mathbf{u}_{R,k}\right)  \label{pert4}%
\end{equation}

We use now \ref{pert1} and drop the second order terms. We obtain
\begin{equation}
\widehat{M}\delta\mathbf{u}_{R,k}+\delta\widehat{M}\mathbf{u}_{R,k}%
=v_{k}\delta\mathbf{u}_{R,k}+\delta\nu_{k}\mathbf{u}_{R,k} \label{pert5}%
\end{equation}

By left multiplication with $\mathbf{u}_{L,k}$ results%
\begin{equation}
\left\langle \mathbf{v}_{L,k},\widehat{M}\delta\mathbf{u}_{R,k}\right\rangle
+\left\langle \mathbf{v}_{L,k},\delta\widehat{M}\mathbf{u}_{R,k}\right\rangle
=v_{k}\left\langle \mathbf{v}_{L,k},\delta\mathbf{u}_{R,k}\right\rangle
+\delta\nu_{k}\left\langle \mathbf{v}_{L,k},\mathbf{u}_{R,k}\right\rangle
\label{pert6}%
\end{equation}

By using Eqs.(\ref{pert0}, \ref{pert2}) \ , the first term in l.h.s. in
Eq.(\ref{pert6}) can be rewritten as $\left\langle \mathbf{v}_{L,k}%
,\widehat{M}\delta\mathbf{u}_{R,k}\right\rangle =\left\langle \widehat{M}%
^{T}\mathbf{v}_{L,k},\delta\mathbf{u}_{R,k}\right\rangle =v_{k}\left\langle
\mathbf{v}_{L,k},\delta\mathbf{u}_{R,k}\right\rangle $ .Combined with
Eq.(\ref{pert6}) we obtain Eq.(\ref{pert3}), that completes the proof.
\end{proof}

\subparagraph{Multiple eigenvalue, semisimple (without Jordan blocks)}

Must of physical cases from this category are related to some discrete or
continuos symmetry group, and the explicit symmetry breaking due to asymmetric
perturbation ( Like Zeeman or Stark effects). \ Let $v_{k}$ an eigenvalue of
multiplicity $m$. \ The subspace of left and right eigenvectors has the same
dimension. Let $\mathbf{v}_{L,k,a}$ respectively $\mathbf{u}_{L,k,a}$ , with
$1\leq a\leq m$ the set of left and right eigenvectors%
\begin{align}
\widehat{M}\mathbf{u}_{R,k,a}  &  =\nu_{k}\mathbf{u}_{R,k,a}\label{pert7}\\
\widehat{M}^{T}\mathbf{v}_{L,k,a}  &  =\nu_{k}\mathbf{v}_{L,k,a} \label{pert8}%
\end{align}

Define the $m\times m$ matrices $\widehat{F}$ and $\widehat{G}$ \ by matrix
elements as follows%
\begin{align}
F_{a,b}  &  =\left\langle \mathbf{v}_{L,k,a},\widehat{M}\mathbf{u}%
_{R,k,b}\right\rangle \label{pert9}\\
G_{a,b}  &  =\left\langle \mathbf{v}_{L,k,a},\mathbf{u}_{R,k,b}\right\rangle
\label{pert10}%
\end{align}

Denote by $\delta\nu_{k,a}=z_{k}$ with $1\leq a\leq m$ , the sequence of $m$
roots $z_{k}$ of the secular equation in $z$%
\begin{equation}
\det\left[  \widehat{F}-z\widehat{G}\right]  =0 \label{pert11}%
\end{equation}

We have the following

\begin{proposition}
\label{PropositionSemisimpleDegenerateEigenval} \ The multiplicity, $m$
semisimple eigenvalue $\nu_{k}$ under the effect of perturbation\ splits into
$m$ eigenvalues
\begin{equation}
\nu_{k}\rightarrow\nu_{k}+\delta\nu_{k,a} \label{pert12}%
\end{equation}

\end{proposition}

\begin{proof}
Combination of the previous proof and from the degenerate case from quantum
mechanical textbooks.
\end{proof}

\subsubsection{The case of diagonal noise}

Our equations will be simplified in the particular case, when the noise
in\ SDE(\ref{60})\ has the form%
\begin{equation}
\rho_{j,a}^{i}=\delta_{j}^{i}r_{a}^{i} \label{280}%
\end{equation}

In this case%
\begin{align}
C_{j,n}^{i,m}  &  =\delta_{j}^{i}\delta_{n}^{m}V_{2}^{i,m}\label{290}\\
V_{2}^{i,m}  &  =\sum_{a=1}^{A}\rho_{j,a}^{i}\rho_{n,a}^{m} \label{300}%
\end{align}

The Eq.(\ref{260}) became%
\begin{align}
\frac{dY_{m}^{k_{1},\cdots,k_{m}}}{dt}  &  =A_{j}^{k_{1}}Y_{m}^{j,k_{2}%
,\cdots,k_{m}}\ +A_{j}^{k_{2}}Y_{m}^{k_{1},j,\cdots,k_{m}}\ +A_{j}^{k_{m}%
}Y_{m}^{k_{1},\cdots,j}\ +\label{310}\\
&
{\displaystyle\sum_{a<b}}
V_{2}^{a,b}Y_{m}^{k_{1},k_{2},\cdots,k_{m}}\nonumber
\end{align}

\section{ Langmuir waves in turbulent background\textbf{ }}

\subsection{Statement of the problem}

The evolution of the electrostatic potential in Langmuir waves in homogenous
finite temperature plasma, after rescaling $\mathbf{x\rightarrow}\sqrt
{3}V_{thermal}\mathbf{x}$ is given by the Klein-Gordon equation%

\begin{align}
\frac{\partial^{2}\phi(\mathbf{x},t)}{\partial t^{2}}  &  =\Delta\phi
-m^{2}\phi\label{L1}\\
m^{2}  &  =\omega_{p,e}^{2}=\frac{n_{e}e^{2}}{\varepsilon_{0}m_{e}}\nonumber
\end{align}

In the case of perturbation of the background electron density $n_{e}%
\rightarrow n_{e}+const\ \zeta(\mathbf{x},t)$ by a temporal white noise
$\ \ $with
\begin{align*}
\mathbb{E}\left[  \zeta(\mathbf{x},t)\right]   &  =0\\
\mathbb{E}\left[  \zeta(\mathbf{x},t)\zeta(\mathbf{x}^{\prime},t^{\prime
})\right]   &  =\delta(t-t^{\prime})C(\mathbf{x},\mathbf{x}^{\prime})
\end{align*}

formally can be written as
\[
\frac{\partial^{2}\phi(\mathbf{x},t)}{\partial t^{2}}=\Delta_{\mathbf{x}}%
\phi(\mathbf{x},t)-\left[  m^{2}+\zeta(\mathbf{x},t)\right]  \phi
\]

\ Or in a more rigorous form, by introducing the Brownian motion
$w(\mathbf{x},t)$ such that formally $\zeta(\mathbf{x},t)=\ \frac{\partial
w_{\omega}(\mathbf{x},t)}{\partial t}$ , we obtain the system of first order
PSDE%
\begin{align}
d\phi_{\omega}(\mathbf{x},t)  &  =p_{\omega}(\mathbf{x},t)dt\label{L2}\\
dp_{\omega}(\mathbf{x},t)  &  =-\widehat{K}_{\mathbf{x}}\phi_{\omega}dt+\phi
dw_{\omega}(\mathbf{x},t) \label{L3}%
\end{align}

where $\widehat{K_{\mathbf{x}}}=m^{2}-\Delta_{\mathbf{x}}$ and $\mathbb{E}%
_{\omega}\left[  w_{\omega}(\mathbf{x},t)w_{\omega}(\mathbf{x}^{\prime
},t^{\prime})\right]  =\min(t,t^{\prime})C(\mathbf{x},\mathbf{x}^{\prime})$
.$\ $ \ 

Denote%
\begin{align*}
Y_{\phi,\phi}(\mathbf{x}_{1},\mathbf{x}_{2},t)  &  =\mathbb{E}_{\omega}\left[
\Phi_{\omega}(\mathbf{x}_{1},t)\Phi_{\omega}(\mathbf{x}_{2},t)\right] \\
Y_{p,\phi}(\mathbf{x}_{1},\mathbf{x}_{2},t)  &  =\mathbb{E}_{\omega}\left[
p_{\omega}(\mathbf{x}_{1},t)\Phi_{\omega}(\mathbf{x}_{2},t)\right] \\
Y_{\phi,p}(\mathbf{x}_{1},\mathbf{x}_{2},t)  &  =\mathbb{E}_{\omega}\left[
\Phi_{\omega}(\mathbf{x}_{1},t)p_{\omega}(\mathbf{x}_{2},t)\right] \\
Y_{p,p}(\mathbf{x}_{1},\mathbf{x}_{2},t)  &  =\mathbb{E}_{\omega}\left[
p_{\omega}(\mathbf{x}_{1},t)p_{\omega}(\mathbf{x}_{2},t)\right]
\end{align*}

By the same procedure we obtain the following system of linear equation for
the second order correlation functions%
\begin{align}
\frac{dY_{\phi,\phi}(\mathbf{x}_{1},\mathbf{x}_{2},t)}{dt}  &  =Y_{p,\phi
}(\mathbf{x}_{1},\mathbf{x}_{2},t)+Y_{\phi,p}(\mathbf{x}_{1},\mathbf{x}%
_{2},t)\label{L3.1}\\
\frac{dY_{p,\phi}(\mathbf{x}_{1},\mathbf{x}_{2},t)}{dt}  &  =-\widehat
{K}_{_{\mathbf{x}_{1}}}Y_{\phi,\phi}(\mathbf{x}_{1},\mathbf{x}_{2}%
,t)+Y_{p,p}(\mathbf{x}_{1},\mathbf{x}_{2},t)\label{L3.2}\\
\frac{dY_{\phi,p}(\mathbf{x}_{1},\mathbf{x}_{2},t)}{dt}  &  =Y_{p,p}%
(\mathbf{x}_{1},\mathbf{x}_{2},t)-\widehat{K}_{_{\mathbf{x}_{2}}}Y_{\phi,\phi
}(\mathbf{x}_{1},\mathbf{x}_{2},t)\label{L3.3}\\
\frac{dY_{p,p}(\mathbf{x}_{1},\mathbf{x}_{2},t)}{dt}  &  =-\widehat
{K}_{_{\mathbf{x}_{1}}}Y_{\phi,p}(\mathbf{x}_{1},\mathbf{x}_{2},t)-\widehat
{K}_{_{\mathbf{x}_{2}}}Y_{p,\phi}(\mathbf{x}_{1},\mathbf{x}_{2},t)\label{L3.4}%
\\
&  +C\left(  \mathbf{x}_{1},\ \mathbf{x}_{2}\right)  Y_{\phi,\phi}%
(\mathbf{x}_{1},\mathbf{x}_{2},t)\nonumber
\end{align}

\bigskip In order to study the $t\rightarrow\infty$ behavior of the solutions
of Eqs. (\ref{L3.1}-\ref{L3.4}), we consider the particular class of
solutions
\begin{equation}
Y_{a,b}(\mathbf{x}_{1},\mathbf{x}_{2},t)=\exp(\lambda t)Z_{a,b}(\mathbf{x}%
_{1},\mathbf{x}_{2});a,b\in\{p,\phi\} \label{L3.5}%
\end{equation}
and after simple algebra we find the following problem:

\begin{problem}
\label{PorblemGeneralizedEigenvalue} Find $\lambda$ with maximal real part
from the generalized eigenvalue problem%
\begin{equation}
\left[  \lambda^{4}+2\lambda^{2}\left(  2m^{2}-\Delta_{\mathbf{x}_{1}}%
-\Delta_{\mathbf{x}_{2}}\mathbf{\ }\right)  +\left(  \Delta_{\mathbf{x}_{1}%
}-\Delta_{\mathbf{x}_{2}}\right)  ^{2}-2\lambda C\left(  \mathbf{x}%
_{1},\ \mathbf{x}_{2}\right)  \right]  Z_{\phi,\phi}(\mathbf{x}_{1}%
,\mathbf{x}_{2})=0 \label{L3.6}%
\end{equation}

\end{problem}

The boundary conditions are obtained in a straightforward way from the
boundary conditions of the original deterministic problem.

\subsubsection{Alternative notations}

A more transparent way is given by a compact notation. We will denote%
\[%
\begin{pmatrix}
\Psi_{1}\\
\Psi_{2}%
\end{pmatrix}
=%
\begin{pmatrix}
\phi(\mathbf{x},t)\\
p(\mathbf{x},t)
\end{pmatrix}
\]

\bigskip

$\ $ Define the \textbf{operator valued }$2\times2$ matrices $\widehat
{\mathbf{L}}_{\mathbf{x}}$ and $\widehat{\mathbf{J}}$\ as follows:%
\begin{align*}
\widehat{\mathbf{L}}^{\mathbf{x}}  &  =%
\begin{pmatrix}
\widehat{0} & \widehat{1}\\
-\widehat{K}_{_{\mathbf{x}}} & \widehat{0}%
\end{pmatrix}
\\
\widehat{\mathbf{J}}  &  =%
\begin{pmatrix}
\widehat{0} & \widehat{0}\\
\widehat{1} & \widehat{0}%
\end{pmatrix}
\end{align*}

or by components

Then Eqs.(\ref{L2}, \ref{L3}) became%
\begin{equation}
d\Psi(\mathbf{x},t)=\widehat{L}^{\mathbf{x}}\Psi dt+\widehat{J}\Psi dw
\label{L4}%
\end{equation}

We use the notation
\begin{equation}
Y_{i_{1},i_{2}}(\mathbf{x}_{1},\mathbf{x}_{2},t)=\mathbb{E}_{\omega}\left[
\Psi_{i_{1},\omega}(\mathbf{x}_{1},t)\Psi_{i_{1},\omega}(\mathbf{x}%
_{2},t)\right]  ;~1\leq i_{1},i_{2}\leq2 \label{L5}%
\end{equation}

Respectively Eqs.(\ref{L3.1}-\ref{L3.4}) became%
\begin{align}
\frac{dY_{i_{1},i_{2}}(\mathbf{x}_{1},\mathbf{x}_{2},t)}{dt}  &
=L_{i_{1},j_{1}}^{\mathbf{x}_{1}}Y_{j_{1},i_{2}}+L_{i_{2},j_{2}}%
^{\mathbf{x}_{2}}Y_{i_{1},j_{2}}\label{L6}\\
&  +J_{i_{1},j_{1}}J_{i_{2},j_{2}}Y_{j_{1},j_{2}}\ C\left(  \mathbf{x}%
_{1},\ \mathbf{x}_{2}\right)  ;1\leq i_{1},i_{2},j_{1},j_{2}\leq2\nonumber
\end{align}

Consider the stability problem of the moments $Y_{i_{1},i_{2}}(\mathbf{x}%
_{1},\mathbf{x}_{2},t)$ . From the representation
\begin{equation}
Y_{i_{1},i_{2}}(\mathbf{x}_{1},\mathbf{x}_{2},t)=\exp\left(  i\omega t\right)
Z_{i_{1},i_{2}}(\mathbf{x}_{1},\mathbf{x}_{2}) \label{L7}%
\end{equation}
and Eq.(\ref{L6}) the following eigenvalue problem results
\begin{align}
\ i\omega Z_{i_{1},i_{2}}(\mathbf{x}_{1},\mathbf{x}_{2})  &  =L_{i_{1},j_{1}%
}^{\mathbf{x}_{1}}Z_{j_{1},i_{2}}+L_{i_{2},j_{2}}^{\mathbf{x}_{2}}%
Z_{i_{1},j_{2}}\label{L8}\\
&  +J_{i_{1},j_{1}}J_{i_{2},j_{2}}Z_{j_{1},j_{2}}\ C\left(  \mathbf{x}%
_{1},\ \mathbf{x}_{2}\right) \nonumber
\end{align}
In analogy to the Eq.(\ref{270}) similar equations for the higher order
correlation functions can be obtained.

\subsection{Particular solutions of the eigenvalue equation (\ref{L3.6},
\ref{L8}).}

\subsubsection{Zero noise limit}

This is useful for small noise limit perturbative calculations.

In line with the classical stability analysis we will denote%
\begin{align}
\varepsilon_{\mathbf{k}}  &  =\sqrt{m^{2}+\mathbf{k}^{2}}\label{L8.1}\\
\lambda &  =i\omega\label{L8.2}%
\end{align}

The classical dispersion relation, for the Langmuir waves is
\begin{equation}
\omega_{\mathbf{k}}=\pm\varepsilon_{\mathbf{k}} \label{L9}%
\end{equation}

When $C=0$ we have the unperturbed eigenvectors in Eq.(\ref{L8}) of the form
\begin{align}
Z_{i_{1},i_{2}}(\mathbf{x}_{1},\mathbf{x}_{2})  &  =v_{i_{1},i_{2}}\exp\left[
i\left(  \mathbf{k}_{1}\mathbf{x}_{1}+\mathbf{k}_{2}\mathbf{x}_{2}\right)
\right] \label{L10}\\
v_{i_{1},i_{2}}  &  =u_{i_{1}}u_{i_{2}} \label{L10.1}%
\end{align}

By using Eqs.(\ref{L8}, \ref{L10}, \ref{L10.1}) after simple algebra we find
that the generalization of the classical dispersion relation, for the case of
\textbf{two point correlation function,} from Eq.(\ref{L9}) is,
\begin{equation}
\omega_{\mathbf{k}_{1},\mathbf{k}_{2}}=\pm\varepsilon_{\mathbf{k}_{1}}%
\pm\varepsilon_{\mathbf{k}_{2}} \label{L11}%
\end{equation}

\subsection{}

\subsubsection{Complete spatial
correlations\label{SubSubSectCompleteSpatialCorrelation} (\cite{SXG}) .}

This part is included only for the purpose of clarification the application of
formalism. Can be considered as a limiting case , of the results exposed in
subsection\textbf{ (\ref{SectionPositiveCorrelations})}

Consider the simplest soluble case $C\left(  \mathbf{x}_{1},\ \mathbf{x}%
_{2}\right)  =\sigma^{2}=const$ . In this case the ansatz \ from
Eq.(\ref{L10}) can be used and from \ Eq.(\ref{L8}) results (see Appendix
\ref{AppendixConstantCorrelation}, Eq.(\ref{apend17}))%
\[
4\varepsilon_{\mathbf{k}_{1}}^{2}\varepsilon_{\mathbf{k}_{2}}^{2}%
+2i\omega\sigma^{2}=\left(  \varepsilon_{\mathbf{k}_{1}}^{2}+\varepsilon
_{\mathbf{k}_{2}}^{2}-\omega^{2}\right)  ^{2}%
\]

or by using Eq.(\ref{L8.2})
\begin{equation}
\lambda^{4}+2\lambda^{2}(\varepsilon_{\mathbf{k}_{1}}^{2}+\varepsilon
_{\mathbf{k}_{2}}^{2})+\left(  \varepsilon_{\mathbf{k}_{1}}^{2}-\varepsilon
_{\mathbf{k}_{2}}^{2}\right)  ^{2}=2\lambda\sigma^{2} \label{L12}%
\end{equation}

From Eq.( \ref{L12}) results that if $|\mathbf{k}_{1}|-|\mathbf{k}_{2}|$ is
sufficiently small then we have two real solutions with $\lambda>0$ , so the
random multiplicative perturbation of the background always produce parametric
destabilization of the $2$ point correlation functions $Y_{\phi,\phi
}(\mathbf{x}_{1},\mathbf{x}_{2},t)=\mathbb{E}_{\omega}\left[  \Phi_{\omega
}(\mathbf{x}_{1},t)\Phi_{\omega}(\mathbf{x}_{2},t)\right]  $. Due to the fact
that $Y_{\phi,\phi}(\mathbf{x}_{1},\mathbf{x}_{2},t)$ is symmetric with
respect permutations $\mathbf{x}_{1\leftrightarrows}\mathbf{x}_{2}$ , the
modes with to the modes (now labelled by a couple $(\mathbf{k}_{1}%
,\mathbf{k}_{2})$) with \ $|\mathbf{k}_{1}|\ \approx|\mathbf{k}_{2}|$ has
non-vanishing contributions. For $|\mathbf{k}_{1}|=|\mathbf{k}_{2}|$ we have
for the dominating mode%
\[
\lambda^{3}+4\lambda\varepsilon_{\mathbf{k}}^{2}=2\ \sigma^{2}%
\]

For high $|\mathbf{k|}$ the the exponential growth is dominated by
$\lambda\underset{|\mathbf{k|\rightarrow\infty}}{\asymp}\sigma^{2}%
/(2\varepsilon_{\mathbf{k}}^{2})+\mathcal{O}(1/\varepsilon_{\mathbf{k}}^{8})$
. The most dangerous modes are those with for $|\mathbf{k}_{1}|\ \approx
|\mathbf{k}_{2}|\approx0$. For small $\sigma$ we have $\lambda_{\max}%
\underset{\sigma\mathbf{\rightarrow0}}{\asymp}\sigma^{2}/(2m^{2}%
)+\mathcal{O}(\sigma^{6})$

\subsubsection{Homogenous system}

Suppose that the random field is homogenous, so $C\left(  \mathbf{x}%
_{1},\ \mathbf{x}_{2}\right)  =C\left(  \mathbf{x}_{1}-\mathbf{x}_{2}\right)
$ , with $C\left(  \mathbf{-x}\right)  =C\mathbf{(x})$. We introduce the
"center of mass" and relative distance coordinates $\mathbf{\ }\left(
\mathbf{x}_{1}+\mathbf{x}_{2}\right)  /2$,$\mathbf{r=x}_{1}-\mathbf{x}_{2}$
and use the ansatz in Eq.(\ref{L3.6})%
\begin{equation}
Z_{\phi,\phi}(\mathbf{x}_{1},\mathbf{x}_{2})=\exp\left[  i\mathbf{k~}\left(
\mathbf{x}_{1}+\mathbf{x}_{2}\right)  /2\right]  \Psi(\mathbf{r})\ \label{L13}%
\end{equation}

The resulting generalized eigenvalue problem is%
\begin{equation}
\left[  \lambda^{4}+2\lambda^{2}\left(  2m^{2}+\mathbf{k}^{2}/2-2\Delta
_{\mathbf{x}}\mathbf{\ }\right)  +4\left(  i\mathbf{k}\nabla_{\mathbf{x}%
}\right)  ^{2}-2\lambda C\left(  \mathbf{x}\right)  \right]  \Psi
(\mathbf{r})=0 \label{L14}%
\end{equation}

\paragraph{Temporal and spatial white noise. \cite{SXG}}

we study now the opposite case studied previously in subsubsection
(\ref{SubSubSectCompleteSpatialCorrelation}). Consider the case one
dimensional version of the previous model, with Dirac-Delta correlation
function $\ C\left(  \mathbf{x}_{1},\ \mathbf{x}_{2}\right)  =\sigma^{2}%
\delta(\mathbf{x}_{1}-\mathbf{x}_{2})$. It is known that the bound state
problem in the Schr\"{o}dinger equation in higher dimension with delta
function potential give mathematically inconsistent results \cite{Atkinson}. \ 

The one dimensional version of Eq. (\ref{L14}) is%
\begin{align}
A(\lambda,k)\psi(x)-B(\lambda,k)\psi^{\prime\prime}(x)-2\lambda\sigma
^{2}\delta(x)\psi(x)  &  =0\\
A(\lambda,k)  &  =\lambda^{4}+4\lambda^{2}m^{2}+\lambda^{2}k^{2}\\
B(\lambda,k)  &  =4(\lambda^{2}+k^{2})
\end{align}

From Eq.(\ref{ad1}) results the continuity and jump conditions for the
restriction of $\psi_{\pm}(x)$ \ to the domains $x>0$ respectively $x<0$.
\begin{align}
\underset{x\nearrow0}{\lim}\psi_{-}(x)  &  =\underset{x\ \searrow0}{\lim}%
\psi_{+}(x)\label{ad4}\\
B(\lambda,k)\left[  \underset{x\ \searrow0}{\lim}\psi_{+}^{\prime
}(x)-\underset{x\nearrow0}{\lim}\psi_{-}^{\prime}(x)\right]  \ +2\lambda
\sigma^{2}\psi_{+}(0)  &  =0 \label{ad5}%
\end{align}

The class of physically admissible (because $\lambda$ is complex) solutions
are of the form%
\begin{align}
\psi_{+}(x)  &  =\exp(-\alpha x)\\
\psi_{-}(x)  &  =\exp(\alpha x)\nonumber\\
\operatorname{Re}(\alpha)  &  \geq0\ \label{ad6}%
\end{align}

By straightforward calculations from Eqs.(\ref{ad4}, \ref{ad5}) results%
\begin{align}
B(\lambda,k)\alpha^{2}  &  =A(\lambda,k)\label{ad}\\
2B(\lambda,k)\alpha &  =2\lambda\sigma^{2} \label{ad7}%
\end{align}

And we obtain
\begin{equation}
\sigma^{4}=4(\lambda^{2}+k^{2})(\lambda^{2}+4m^{2}+k^{2}) \label{ad8}%
\end{equation}

From Eqs.(\ref{ad7}, \ref{ad3}) results that for any root of Eq.(\ref{ad8})
with $\operatorname{Re}(\lambda)\geq0$ \ , the condition from Eq.(\ref{ad6})
is fulfilled.

\ It is clear that now there is \ a threshold: when $\sigma^{4}<4k^{2}%
(4m^{2}+k^{2})$ then $\lambda$ is imaginary. For $\sigma^{4}>4k^{2}%
(4m^{2}+k^{2})$ there is a real positive $\lambda$, so the system is unstable.
We remark an important mechanism: for low noise intensity first the large
wavelength structures are destabilized. So by this mechanism there is transfer
of energy from short to the long wavelength. The most sensitive to the
destabilization are the mode with small $k$.

\paragraph{Positive correlations \label{SectionPositiveCorrelations}.}

Consider the solutions of the Eq.(\ref{L14}) \ in the case \ $C(\mathbf{x}%
)\geq0$ and $\ \underset{|\mathbf{x}|\rightarrow\infty}{\lim}C(\mathbf{x})=0$. \ 

\subparagraph{Case $k=0$}

We are interested in the case $\lambda\neq0$, so the generalized eigenvalue
problem can be reformulated as follows.
\begin{equation}
-\ \Delta_{\mathbf{x}}\mathbf{\ }\Psi(\mathbf{r})+\left[  \ -\frac{1}%
{2\lambda}C\left(  \mathbf{x}\right)  \right]  \Psi(\mathbf{r})=-(m^{2}%
+\frac{\lambda^{2}}{4})\Psi(\mathbf{r}) \label{L15}%
\end{equation}

When $0<\lambda<+\infty$ , ,because $C>0$, the lowest eigenvalue $E_{\lambda}$
of the eigenvalue problem$-\ \Delta_{\mathbf{x}}\mathbf{\ }\Psi(\mathbf{r}%
)+\left[  \ -\frac{1}{2\lambda}C\left(  \mathbf{x}\right)  \right]
\Psi_{0,\lambda}(\mathbf{r})=E_{\lambda}\Psi_{0,\lambda}(\mathbf{r})$%
\begin{align}
-\ \Delta_{\mathbf{x}}\mathbf{\ }\Psi(\mathbf{r})+\left[  \ -\frac{1}%
{2\lambda}C\left(  \mathbf{x}\right)  \right]  \Psi_{0,\lambda}(\mathbf{r})
&  =E_{\lambda}\Psi_{0,\lambda}(\mathbf{r})\label{L16}\\
-(m^{2}+\frac{\lambda^{2}}{4})  &  =E_{\lambda}%
\end{align}

is a monotone increasing continuos function, from $-\infty$ to $0$.

Indeed, the term $\ -\frac{1}{2\lambda}C\left(  \mathbf{x}\right)  $ \ is like
an attractive potential well in a ground state problem in the Eq.(\ref{L16}),
formally identical to Schr\"{o}dinger equation. For $\lambda\rightarrow0$ it
is very deep potential well, can be proven rigorously that $E_{\lambda
}\underset{\lambda\rightarrow0}{\lim}E_{\lambda}=-\infty$ . When $\lambda$
increases then there is a critical value $\lambda_{crit}>0$ such that the
bound state problem from Eq.(\ref{L16}) has no more solutions . So we have
$\underset{\lambda\ \nearrow\lambda_{crit}}{\lim}E_{\lambda}=0$. (Previous
asymptotic estimates can be verified by approximating the potential well with
a rectangular one, at least in the case of fast correlation decay).

Results that there is a one and only one value of $\lambda>0$ such that
$E_{\lambda}=-(m^{2}+\frac{\lambda^{2}}{4})$ (because the r.h.s. is
decreasing, from $-m^{2}$ to $\ -\infty$) .

\begin{conclusion}
It follows that there exists at least an eigenvalue $\lambda>0$ in the
generalized eigenvalue problem from Eq.(\ref{L15}), so at least the mode
$\ k=0$ is destabilized for arbitrary small noise intensity.
\end{conclusion}

\subparagraph{Case $|\mathbf{k}|\rightarrow\infty$.}

In \ general, Eq.(\ref{L14}), by suitable choice of the coordinate system
($\mathbf{k\parallel}Oz$), can be reformulated as follows.
\begin{align}
-\left[  \Delta_{\perp}+\left(  1+\left(  \frac{k}{\lambda}\right)
^{2}\right)  \frac{\partial^{2}}{\partial z^{2}}\right]  \Psi_{0,\lambda
}(\mathbf{r})+\left[  \ -\frac{1}{2\lambda}C\left(  \mathbf{x}\right)
\right]  \Psi_{0,\lambda}(\mathbf{r})  &  =E_{\lambda}\Psi_{0,\lambda
}(\mathbf{r})\label{L17}\\
-(m^{2}+\frac{k^{2}}{4}+\frac{\lambda^{2}}{4})  &  =E_{\lambda} \label{L18}%
\end{align}

where $\Delta_{\perp}=\frac{\partial^{2}}{\partial x^{2}}+\frac{\partial^{2}%
}{\partial y^{2}}$. \ Because of the new term $\left(  \frac{k}{\lambda
}\right)  ^{2}\frac{\partial^{2}}{\partial z^{2}}$ , the ground state energy
has a lower bound.

\ Be denoting with $\ \delta$ the typical width of the maximum of $C\left(
\mathbf{x}\right)  $ near $\mathbf{x}=0$, respectively $\ c=C(0)$, for large
$\ k$ \ and small $\lambda$ we have the estimate $E_{\lambda}>\mathcal{O(-}%
\frac{C(0)^{2}\delta}{k^{2}}\mathcal{)}$ . So for fixed large values of $\ k$
there is no real $\lambda$ that satisfies Eq.(\ref{L18}) .

\begin{conclusion}
. For $k>0$ there is a threshold for the intensity of the multiplicative noise
such that the mode remains stable below this threshold. 
\end{conclusion}

\section{Conclusions.}

Similar to the simplest one dimensional case exposed in the works
\cite{Lindberg}, \cite{Gitterman}, the   Langmuir waves in a turbulent
background are destabilized. The energy of the fluctuations related to the
multiplicative noise is transferred with an exponential rate to the long
wavelength fluctuations preferentially. In contradistinction with one
dimensional case, there is a threshold\ in the noise intensity for each mode,
such that below this threshold the mode remains stable. 

\section{Appendix \label{AppendixConstantCorrelation}}

Denote the tensor with components $v_{i_{1},i_{2}}$ from Eq.(\ref{L10}) by
$\mathbf{v}$ . We use the notation Eq.(\ref{8.1}) and define $2\times2$
matrices \ $\widehat{E}(\mathbf{k)},\widehat{\mathbf{a}}$ as follows
\begin{align}
\ \widehat{E}(\mathbf{k)}  &  \mathbf{=}%
\begin{pmatrix}
0 & 1\\
-\varepsilon_{\mathbf{k}}^{2} & 0
\end{pmatrix}
\label{apend1}\\
\widehat{\mathbf{a}}  &  =%
\begin{pmatrix}
0 & 0\\
1 & 0
\end{pmatrix}
\label{apend2}\\
\widehat{\mathbf{u}}_{1}  &  =%
\begin{pmatrix}
1 & 0\\
0 & 0
\end{pmatrix}
\label{apend2.1}\\
\widehat{\mathbf{u}}_{2}  &  =%
\begin{pmatrix}
0 & 0\\
0 & 1
\end{pmatrix}
\label{apend2.2}%
\end{align}

and by $\widehat{\mathbf{1}}\ $ \ the $2\times2\ \ $unit matrix. The
eigenvalue equation (\ref{L8}) with ansatz Eq.(\ref{L10}) became%
\begin{equation}
i\omega\mathbf{v=}\widehat{B}\mathbf{v} \label{append3}%
\end{equation}

where
\begin{equation}
\widehat{B}=\widehat{E}(\mathbf{k}_{1}\mathbf{)\otimes\widehat{\mathbf{1}%
}\ +\widehat{\mathbf{1}}\ \otimes}\widehat{E}(\mathbf{k}_{2}\mathbf{)}%
+\sigma^{2}\mathbf{\widehat{a}\otimes}\widehat{\mathbf{a}} \label{apend4}%
\end{equation}

From Eq. (\ref{append3}) results
\begin{equation}
-\omega^{2}\mathbf{v=}\widehat{B}^{2}\mathbf{v} \label{apend5}%
\end{equation}

We will use the following matrix identities, that can be verified easily%
\begin{align}
\widehat{E}(\mathbf{k)}^{2}  &  =-\varepsilon_{\mathbf{k}}^{2}%
~\mathbf{\widehat{\mathbf{1}}\ ;~~\widehat{a}}^{2}=0\label{apend6}\\
\widehat{E}(\mathbf{k)~\widehat{a}}  &  \mathbf{=}\widehat{\mathbf{u}}%
_{1}~;\mathbf{\widehat{a}~}\widehat{E}(\mathbf{k)=}\widehat{\mathbf{u}}%
_{2}\label{apend7}\\
\widehat{\mathbf{u}}_{1}+\widehat{\mathbf{u}}_{2}  &  =\mathbf{\widehat
{\mathbf{1}}\ } \label{apend8}%
\end{align}

From Eqs.(\ref{apend4}, \ref{apend6}, \ref{apend7}, \ref{apend8}) results%
\begin{align}
\widehat{B}^{2}  &  =-\left(  \varepsilon_{\mathbf{k}_{1}}^{2}+\varepsilon
_{\mathbf{k}_{2}}^{2}\right)  \mathbf{\widehat{\mathbf{1}}\ \otimes
\widehat{\mathbf{1}}}+2\widehat{E}(\mathbf{k}_{1}\mathbf{)\otimes\widehat
{E}(\mathbf{k}_{2}\mathbf{)}\ +\ }\label{apend9}\\
&  \sigma^{2}\left(  \mathbf{\widehat{\mathbf{1}}\ \otimes\widehat{a}%
+\widehat{a}\ \otimes\widehat{\mathbf{1}}}\right)
\end{align}

and \ with the notation \
\begin{equation}
y=\left(  \varepsilon_{\mathbf{k}_{1}}^{2}+\varepsilon_{\mathbf{k}_{2}}%
^{2}\right)  -\omega^{2} \label{apend10}%
\end{equation}

the eigenvalue equation(\ref{apend5}) became%
\begin{equation}
\widehat{C}\mathbf{v}=y\mathbf{v} \label{apend11}%
\end{equation}

with%
\begin{equation}
\widehat{C}=2\widehat{E}(\mathbf{k}_{1}\mathbf{)\otimes\widehat{E}%
(\mathbf{k}_{2}\mathbf{)}\ +}\sigma^{2}\left(  \mathbf{\widehat{\mathbf{1}%
}\ \otimes\widehat{a}+\widehat{a}\ \otimes\widehat{\mathbf{1}}}\right)
\label{apend12}%
\end{equation}

It follows
\begin{equation}
\widehat{C}^{2}\mathbf{v}=y^{2}\mathbf{v} \label{apend13}%
\end{equation}

and from Eq.(\ref{apend12}, \ref{apend6}-\ref{apend8},) we obtain%
\begin{equation}
\widehat{C}^{2}=4\varepsilon_{\mathbf{k}_{1}}^{2}\varepsilon_{\mathbf{k}_{2}%
}^{2}+2\sigma^{2}\left[  \widehat{E}(\mathbf{k}_{1}\mathbf{)\otimes
\widehat{\mathbf{1}}\ +\widehat{\mathbf{1}}\ \otimes}\widehat{E}%
(\mathbf{k}_{2}\mathbf{)}+\sigma^{2}\mathbf{\widehat{a}\otimes}\widehat
{\mathbf{a}}\right]  \mathbf{\ } \label{apend14}%
\end{equation}

or compared with Eq.(\ref{apend4})
\begin{equation}
\widehat{C}^{2}=4\varepsilon_{\mathbf{k}_{1}}^{2}\varepsilon_{\mathbf{k}_{2}%
}^{2}+2\sigma^{2}\widehat{B} \label{apend15}%
\end{equation}

So, from Eqs. (\ref{apend13}, \ref{apend15}) we obtain
\begin{equation}
\left[  4\varepsilon_{\mathbf{k}_{1}}^{2}\varepsilon_{\mathbf{k}_{2}}%
^{2}+2\sigma^{2}\widehat{B}\right]  \mathbf{v}=y^{2}\mathbf{v} \label{apend16}%
\end{equation}

or, combined with Eqs.(\ref{append3}, \ref{apend10}) we obtain the relation%
\begin{equation}
4\varepsilon_{\mathbf{k}_{1}}^{2}\varepsilon_{\mathbf{k}_{2}}^{2}%
+2i\omega\sigma^{2}=\left(  \left(  \varepsilon_{\mathbf{k}_{1}}%
^{2}+\varepsilon_{\mathbf{k}_{2}}^{2}\right)  -\omega^{2}\right)
^{2}\label{apend17}%
\end{equation}

\end{document}